\documentclass[11pt]{article}

\usepackage{mathtools}
\usepackage{fullpage}
\usepackage{amssymb,amsmath,amsfonts,bm,epsfig,graphicx,theorem,latexsym}
\usepackage{rotating,setspace,latexsym,epsf,color,subfigure}
\usepackage{cite}
\usepackage{caption}
\usepackage{hyperref}
\usepackage{url}
\usepackage[usenames,dvipsnames]{xcolor}
\usepackage{amstext}
\usepackage{verbatim}
\usepackage{paralist}
\usepackage{ulem}
\usepackage[numbers]{natbib}
\usepackage{todonotes}

\newcommand{\wellb}{\sf{Behaved}}
\newcommand{\sufx}{\overline{\mathbf{C}}}
\newcommand{\bc}{\overline{\sf{\bf{c}}}}

\newcommand{\ms}{\mathcal{S}}
\newcommand{\ml}{\mathcal{L}}
\newcommand{\mc}{\mathcal{C}}
\newcommand{\mx}{\mathcal{X}}
\newcommand{\mt}{\mathcal{T}}
\newcommand{\my}{\mathcal{Y}}
\newcommand{\mv}{\mathcal{V}}
\newcommand{\mU}{\mathcal{U}}
\newcommand{\advi}{\mbox{Adv}^i}
\newcommand{\advn}{\mbox{Adv}^n}
\newcommand{\adv}{\mathbf{Adv}}
\newcommand{\enc}{{\sf{Enc}}}
\newcommand{\dec}{{\sf{Dec}}}
\newcommand{\ce}{{\sf{CE}}}
\newcommand{\pr}{\mathbb{P}}

\newcommand{\bx}{{\mathbf{x}}}
\newcommand{\bbx}{{\mathbf{X}}}
\newcommand{\by}{{\mathbf{y}}}
\newcommand{\bby}{{\mathbf{Y}}}
\newcommand{\bz}{{\mathbf{z}}}

\newcommand{\ei}{{\sf{Error}}_{\mbox{I}}\left(u, \by_1, S; {\enc}\right)}

\newcommand{\peI}{P_{{\sf{type}\mbox{-I}}}^n\left(u,~\by_1\right)}

\newcommand{\bit}[1]{{\{0,1\}^{#1}}}

\newtheorem{claim}{Claim}

\newtheorem{theorem}{Theorem}
\newtheorem{prop}{Proposition}
\newtheorem{lemma}{Lemma}

\newtheorem{Definition}{Definition}[section]
\newenvironment{proof}{\noindent{\em Proof:}\hspace*{1em}}{\hfill$\Box$}

\begin{document}

\title{Causal Erasure Channels}

\author{Raef Bassily\footnotemark[1]\thanks{Computer Science and Engineering Department, The Pennsylvania State University, University Park, PA. \texttt{\{bassily,asmith\}@psu.edu}. Supported by National Science Foundation awards \#0941553 (CDI) and \#0747294 (PECASE). } \and Adam Smith\footnotemark[1]}

\date{}
\maketitle

\begin{abstract}
We consider the communication problem over binary \emph{causal adversarial erasure} channels. Such a channel maps $n$ input bits to $n$ output symbols in $\{0,1,\wedge\}$, where $\wedge$ denotes erasure. The channel is \emph{causal} if, for every $i$, the channel adversarially decides whether to erase the $i$th bit of its input based on inputs $1,...,i$, \emph{before} it observes bits $i+1$ to $n$. Such a channel is $p$-bounded if it can erase at most a $p$ fraction of the input bits over the whole transmission duration. Causal channels provide a natural model for channels that obey basic physical restrictions but are otherwise unpredictable or highly variable. For a given erasure rate $p$, our goal is to understand the optimal rate (the ``capacity'') at which a randomized (i.e., stochastic) encoder/decoder can transmit reliably across all causal $p$-bounded erasure channels.

In this paper, we introduce the causal erasure model and provide new upper bounds (impossibility results) and lower bounds (analyses of codes) on the achievable rate. Our bounds separate the achievable rate in the causal erasures setting from the rates achievable in two related models: random erasure channels (strictly weaker) and fully adversarial erasure channels (strictly stronger). Specifically, we show:
\begin{itemize}

\item A strict separation between random and causal erasures for all constant erasure rates $p\in(0,1)$. In particular, we show that the capacity of causal erasure channels is $0$ for $p\geq 1/2$ (while it is nonzero for random erasures).

\item A strict separation between causal and fully adversarial erasures for $p\in(0,\phi)$ where $\phi \approx 0.348$.

\item For $p\in[\phi,1/2)$, we show codes for causal erasures that have higher rate than the best known constructions for fully adversarial channels.
\end{itemize}

Our results contrast with existing results on correcting causal \emph{bit-flip} errors (as opposed to erasures) \cite{DJL08,LJD,HL,DJLS,DeyJL13}. For the separations we provide, the analogous separations for bit-flip models are either not known at all or much weaker.
\end{abstract}


\newcommand{\mypar}[1]{\bigskip \noindent{\bf {#1}.}}

\section{Introduction}

Reliable communication over erasure channels is a central topic in
coding and information theory. Erasure channels are noisy channels
in which symbols are either transmitted intact or ``erased'', that is,
replaced by a special symbol $\wedge$ denoting a visible error. They
are interesting in their own right (in settings where
transmission errors are detectable by the decoder), and as intermediate
abstractions in the construction of codes for other models.

The two classic approaches model erasure channels either as a
known stochastic process (c.f. Shannon\cite{shannon}), or as an
adversarial process subject only to a limit on the \emph{number}
of erasures it can introduce
(c.f. Hamming~\cite{hamming}). Adversarial models are more flexible, as
they capture varying or poorly understood channels. Yet the
maximum rate of reliable transmission over adversarial channels is
much lower than over stochastic channels with a similar rate of erasures.

In this paper, we introduce and study \emph{causal} adversarial
erasure channels. Such channels are adversarial, but limited to
introduce erasures online as the symbols are transmitted, based
only on the symbols sent so far. They provide a natural, intermediate
model between stochastic and fully adversarial models. In particular,
they capture any physical channel over which symbols are sent and
received sequentially. Examples include i.i.d. erasures, as well as a
large range of more complex channels (e.g., burst erasures). We prove
that the achievable rate of causal erasure channels lies strictly
between the achievable rates of analogous stochastic and fully
adversarial models. Our model is inspired by recent work on causal
\emph{error} channels \cite{DJL08,LJD,HL,DJLS,DeyJL13}, discussed in ``Previous Work'', below.

Specifically, an erasure channel is a randomized map from $\bit{n}$ to
$\{0,1,\wedge\}^n$.  The channel is causal if, for every $i$, the
channel  decides whether to erase the $i$th bit of its
input based on inputs $1,...,i$, \emph{before} it observes bits $i+1$
to $n$. The channel is $p$-bounded if it can erase at most $pn$
fraction of the input bits over the whole transmission duration. A
(stochastic) \emph{code} is a pair of (randomized) encoding/decoding
algorithms $(\enc,\dec)$ that map a message space $\mU$ to a codeword
in $\bit{n}$, and a received word in $\{0,1,\wedge\}^n$ to a candidate
message in $\mU$. Given $p$, the code is required
to transmit reliably (with high probability) across \emph{all}
$p$-bounded causal erasure channels. In particular, the channel's
behavior may depend on the code itself, and no secret
randomness is allowed to be shared between the encoder and decoder.
The rate of the code is $\log(|\mU|)/n$ (the ratio
of bits transmitted to channel uses).
We are interested in the
\emph{capacity} $C_p$ of causal erasure channels, that is, the maximum
achievable rate in the limit of large $n$, as a function of $p$. (See
``System Model'', below, for precise definitions.)

Causal channels provide a natural model for channels that
obey basic physical restrictions but are otherwise unpredictable or
highly variable. A line of recent work discussed below (``Previous
Work'') considers causal \emph{bit-flip} channels; ours is the first
to study causal erasures.

We provide new upper bounds (impossibility results) and lower bounds
(analyses of codes) on the achievable rate of codes for causal erasures.
To frame our results, consider the two other classes of $p$-bounded
channels mentioned above, namely random erasures and fully adversarial
erasures. The channel that erases a \emph{uniformly random} set of $pn$ positions (or,
essentially equivalently, erases each symbol independently with
probability $p$). The capacity of this channel is $1-p$ (and efficient
constructions are known that achieve this rate). In contrast, the best achievable rate over fully
adversarial channels is less well understood. In terms of asymptotic
rate, codes for $p$-bounded fully adversarial channels are equivalent
to codes in which every pair of valid codewords differ in at least
$pn+1$ positions.\footnote{This equivalence is trivial if we insist
  that the code have \emph{zero} probability of error over fully
  adversarial channels. The equivalence is nontrivial (but still holds, by a
  method-of-expectations argument) if the encoder/decoder can be randomized and a
  small probability of decoding error is allowed.}
Understanding the rate of such
codes is a long-standing open problem in coding theory; the best upper
bounds (a combination of the Bassalygo-Elias \cite{BE} and LP bounds~\cite{LP})
and lower bounds (given by the  Gilbert-Varshamov bound)  are
plotted in Figure~\ref{fig:bounds}.  A few features stand out: the
asymptotic achievable rate over fully adversarial channels is 0 for
$p\geq 1/2$, and the curve has unbounded slope as it approaches $p=0$
(specifically, the maximum rate is $1-\Theta(p\ln(\frac 1 p))$ as $p$ goes to 0).

\begin{figure}[t]
\centerline{\includegraphics[width=0.75\textwidth]{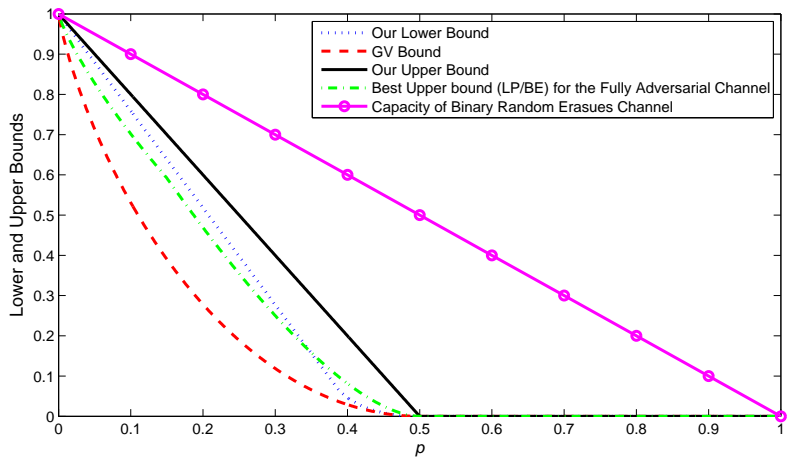}}\caption{Upper and lower bounds on the capacity of the binary-erasure channel.}\label{fig:bounds}
\end{figure}

\mypar{Our Results}
We give two main bounds on the capacity $C_p$ of $p$-bounded causal erasure
channels (depicted in Fig.~\ref{fig:bounds}). We show:
\begin{enumerate}
\item The capacity $C_p$ is at most $(1-2p)^{+}$. This is
  the same value as the Plotkin bound for codes with minimum distance
  $pn+1$, but it requires a different proof; see ``Techniques'', below.

  We show this by giving a particular adversarial strategy, analogous
  to the ``Wait and Push'' strategy of Dey et al. \cite{DJLS}  in the bit-flip setting.

\item The capacity $C_p$ is at least
the function $R_L(p)$ given in Theorem~\ref{lower_bound} and plotted in Fig.~\ref{fig:bounds}.
We show this via a random coding argument inspired by (but quite
different from) that of \cite{HL}. The resulting
encoder/decoder pair are not polynomial time in general.
\end{enumerate}

\bigskip
Our bounds have several implications for the relation between
random, causal, and fully-adversarial erasure models.

\begin{itemize}

\item   For every constant $p\in(0,1)$, the achievable rate of codes
  for causal channels is strictly worse than the rate of codes for
  random errors (since $1-2p< 1-p$). In particular, the achievable rate over causal
  channels is 0 for $p\geq 1/2$ (whereas it is nonzero for random
  errors).

\item The capacity of causal erasure channels is strictly greater than
  that of fully adversarial erasures for  $p\in(0,\phi)$ where $\phi
  \approx 0.348$. This is the point where our lower bound intersects
  the best known upper bounds on the rate of codes with minimum
  distance $pn+1$ (see Figure~\ref{fig:bounds}).

  Moreover, the graph of our lower bound has finite slope at $p=0$,
  meaning that for low erasure rates, $O(pn)$ bits of redundancy
  suffice to tolerate causal erasures, while fully adversarial ones
  require $\Theta(n p \ln(\frac 1 p))$ bits of redundancy.

\item For $p\in[\phi,\frac12)$, we show codes for causal erasures that
  have higher rate than the best known constructions for fully
  adversarial channels. That is,  our lower bound lies strictly above
  the Gilbert-Varshamov bound for all $p$ in $(0,\frac 12)$. Our lower
  bound lies below the best upper bounds for $p\in(\phi,\frac 12)$, however, and a strict
  separation in that range remains an open question.
\end{itemize}

\mypar{Previous Work} A number of works have sought to find middle
ground between the optimism of random-error models and more
pessimistic fully adversarial, ``combinatorial'' error models. For
example, arbitrarily varying channels \cite{ahlswede,CN88a} allow the each symbol to be
corrupted by one of several operators (selected
adversarially). Computationally-bounded channels \cite{lipton,MPSW,GuruswamiS10} consider
channels whose action can be described by a low-complexity circuit.

Most relevant to this work, Dey et al.~\cite{DJL08,LJD,DJLS} and Haviv and
Langberg~\cite{HL} recently studied causal (or ``online'')
\emph{bit-flip} channels (as well as errors over larger alphabets~\cite{DeyJL13}). They describe upper and lower bounds on the
capacity of such channels, and our work was inspired by their
approaches. As in the case of erasures, there are three natural,
nested models
for bit-flip errors: random, causal and fully adversarial.

Our results paint a much more complete picture of the situation for
causal erasures than is known for causal bit-flip errors.
For each of the separations we show, the analogous separation for
bit-flip errors is either not known or much weaker.
\begin{itemize}
\item A strict separation between causal and random bit-flip errors is
  not known to hold for all error rates; for small error rates (less than
  about 0.08), the best-known upper bound on causal errors is the
  capacity of the binary symmetric (random bit-flip error) channel~\cite{DJLS}.
  \item No strict separation is known between causal errors and fully
    adversarial errors. In fact, it is only for a small range of error
    rates that any codes are known to beat the Gilbert-Varshamov bound~\cite{HL}.
\end{itemize}

One may view our results on erasures as an indication that the
separations among bit-flip error models are, in fact, strict. We hope
that our results provide some insight into these questions.

\mypar{Techniques} As mentioned above, the proofs of our upper and
lower bounds are inspired by techniques of \cite{DJLS,HL}. This is
natural, since any
erasure channel can be converted to a bit-flip channel by replacing
erasures with random bits. Upper bounds on erasure channels thus imply upper bounds
on bit-flip channels (and vice-versa for lower bounds). Our results
are much stronger, however, than what follow that way from previous work.

Our upper bound is a strengthening
of one of the bounds of \cite{DJLS} (and of the Plotkin bound).
The main technical innovation is in the lower bound, the heart of which is a bound on the size of a
``forbidden ball'' (a set of points around a codeword within which
the presence of other codewords may cause a decoding error).  The
geometry of this ball is quite different from the analogous structure
for bit-flip errors, and the proof ends up being highly specific to
erasures. 


\section{System Model}\label{system}

We consider communication problem over the class of causal erasure adversarial channels with parameter $p\in[0,1]$, denoted by $\ce_p$. For a transmission duration of $n$ symbols, a channel $\adv\in\ce_p$ is defined by the triple $(\mx^n,\{\advi,i=1,..,n\},\my^n)$ where $\mx=\{0,1\},~\my=\{0,1,\wedge\}$ is the input, output alphabet per symbol, and $\advi:\mx^i\times\my^{i-1}\rightarrow\my,~i=1,...,n$ is a (randomized) function that, at time instant $i$, maps the observed sequence of input symbols up to the current instant $i$ , $x_1^i\in\mx^i$, together with the sequence of all the previous output symbols up to $i-1$, $y_{1}^{i-1}\in\my^{i-1}$, to an output symbol $y_i\in\{x_i,\wedge\}$ such that by the end of transmission, i.e., when $i=n$, the number of erased symbols in $\by$ is at most $pn$. 
Except for the causality constraint and the constraint on the total number of erasures, the channel's behavior is arbitrary.

The transmitter's message set is denoted by $\mU=\{0,1\}^{\lceil nR\rceil}$ for some $R\geq 0$. For simplicity of notation, we assume, w.l.o.g., that $nR$ is an integer. In this paper, we consider two different settings for the message to be transmitted: in Section~\ref{UPPER} where we derive an upper bound on the capacity of $\ce_p$, we will assume a uniformly distributed message $U$ over the set $\mU$ whereas, in Section~\ref{LOWER} where we derive a lower bound on the same capacity, we will assume that the message is arbitrarily fixed and even known to channel before the transmission starts and hence, our construction works for any message $u\in\mU$. Adopting these two different settings in the upper and lower bounds is meant to give stronger results. That is, an upper bound for the uniform message setting implies the same upper bound for any distribution over the message set. On the other hand, a lower bound for the ``worst case'' setting where the message is arbitrarily chosen and known to the channel beforehand implies the same lower bound for any distribution over the message set.


A $(2^{nR},n)$ code is defined as a pair $(\enc,\dec)$ where $\enc:\mU\rightarrow\mx^n$ is a (stochastic) encoder and \mbox{$\dec:\my^n\rightarrow\mU$} is a decoder. No shared randomness is assumed between the encoder and the decoder. The encoder maps (with the possible use of local randomness) a message $U\in\mU$ to a codeword $\bbx\in\mx^n$ which serves as an $n$-bit input of $\adv$. With no loss of generality, $\enc$ is assumed to be injective. That is, for every distinct pair of messages $u,~u'$, we have $\enc(u)\neq\enc(u')$ with probability $1$. The decoder $\dec$ maps the channel's output sequence $\bby\in\my^n$ (with at most $pn$ erasures) to an estimate of the transmitted message $\hat{U}$.

Since we assume different settings of the message distribution in Sections~\ref{UPPER} and \ref{LOWER}, we will have two different versions of the error criterion. In Section~\ref{UPPER}, since the message $U$ is assumed to be uniformly distributed over $\mU$, our error criterion will be the average probability of decoding error. With respect to a code $(\enc, \dec)$ and a channel $\adv\in\ce_p$, the average probability of decoding error denoted by $P_{avg}^{n}\left(\enc,\dec,\adv\right)\triangleq\pr\left(\hat{U}\neq U\right)$ is given by
\begin{align}
&P_{avg}^n\left(\enc,\dec,\adv\right)\triangleq\frac{1} {2^{nR}}\sum_{u\in\mU}\sum_{\bx\in\mx^n}\pr\left(\enc(u)=\bx\right)\sum_{\by\in\my^n}\pr\left(\adv(\bx)=\by\right)\pr\left(\dec(\by)\neq u\right)\label{P_avg-e}
\end{align}
where $\pr\left(\enc(u)=\bx\right)$ is the conditional probability that the output of the encoder $\enc$ is $\bx\in\mx^n$ when the input message is $u\in\mU$, $\pr\left(\advn(\bx)=\by\right)$ is the conditional probability that the channel $\adv$, after the whole transmission duration, outputs the sequence $\by\in\my^n$ given that the input sequence is $\bx\in\mx^n$, and $\pr\left(\dec(\by)=\hat{u}\right)$ is the conditional probability that the output of the decoder $\dec$ is $\hat{u}\in\mU$ given that its input (the received sequence) is $\by\in\my^n$. In Section~\ref{LOWER}, since we consider the setting where the message is arbitrarily fixed and known to the channel in advance, our error criterion will be the maximum probability of decoding error over all messages $u\in\mU$ (i.e., the worst-case probability of error with respect to the set of all messages). With respect to a code $(\enc, \dec)$ and a channel $\adv\in\ce_p$, the maximum probability of decoding error over all the messages, denoted by $P_{max}^n\left(\enc,\dec,\adv\right)$, is given by
\begin{align}
&P_{max}^n\left(\enc,\dec,\adv\right)\triangleq\max_{u\in\mU}\sum_{\bx\in\mx^n}\pr\left(\enc(u)=\bx\right)\sum_{\by\in\my^n}\pr\left(\advn(\bx)=\by\right)\pr\left(\dec(\by)\neq u\right)\label{P_max-e}
\end{align}
When $\enc,~\dec,$ and $\adv$ are clear from the context, we will drop them from the above notation and use just $P_{avg}^n$ (or $P_{max}^n$).


A rate $R$ is said to be achievable for $\ce_p$ if for every $\epsilon>0$ there exists a sequence of codes $\{(2^{n(R-\epsilon)},n):~n\geq1\}$ such that for every $\beta>0$ there exists an integer $n_{\beta}$ such that for all $\adv\in\ce_p$, we have $P_{max}^n<\beta$ for all $n>n_{\beta}$. Note that since the condition on $P_{max}^n$ must hold for every $\adv\in\ce_p$, $\adv$ is  allowed to depend on the code. The capacity of $\ce_p$, denoted as $C_p$, is defined as the supremum of all achievable rates for $\ce_p$.

As a remark on notation, we will use upper-case letters for random variables, lower-case letters for fixed realizations, bold-face letters for vectors, and normal letters for scalars. We will also use $u$ to denote a message


\section{Upper Bound}~\label{UPPER}
Our upper bound on $C_p$, denoted by $R_{\sf{Upper}}(p)$, is formally stated in the following theorem.
\begin{theorem}\label{upper_bound}
For every $p\in[0,1]$, the capacity of $\ce_p$, $C_p$, is at most
\begin{align}
C_p\leq R_{\sf{Upper}}(p)&\triangleq (1-2p)^{+}
\end{align}
where $(x)^+=\max(x,0)$ for $x\in\mathbb{R}$.
\end{theorem}

To prove this upper bound on $C_p$, we show that there exists an adversarial strategy run by some $\adv\in\ce_p$ that imposes a decoding error with a probability bounded away from zero for any $(2^{nR},n)$ code with $R>R_{\sf{Upper}}(p)$. In particular, we show that for every $\epsilon>0$, $R_{\sf{Upper}}(p)+\epsilon$ is not achievable for $\adv\in\ce_p$, i.e., no matter what $(2^{n(R_{\sf{Upper}}(p)+\epsilon)},n)$ code is used or how large $n$ is, there is an adversarial strategy that causes $P_{avg}^n$ to be bounded from below by a positive constant that does not depend on $n$. We assume here a uniformly distributed message $U$ over $\mU$ as discussed in the previous section. Note that this immediately implies the same result if $P_{max}^n$ criterion is used instead and hence our result is even stronger.

The adversarial strategy $\adv$ used is quite similar to the one proposed in \cite{DJLS}. Our adversarial strategy is a ``wait-push'' strategy where the channel (i.e., the adversary) splits the transmission in two phases: (i) the \emph{wait phase}, where the channel observes a prefix $\bx_1$ of length $\ell$ bits (to be specified later) of the transmitted codeword $\bx$ without erasing any bits in this phase. The channel uses this phase to construct a list $\ml_{\bx_1}$ (whose size is potentially smaller than size of the whole code) of ``candidate'' codewords that are consistent with the observed prefix $\bx_1$ among which is the actual codeword chosen by the transmitter, (ii) the \emph{push phase}, where the channel chooses a codeword $\bx'$ randomly from $\ml_{\bx_1}$ (which, with a positive probability, corresponds to a different message than the one originally chosen by the transmitter), then for the last $n-\ell$ bits of the transmission, i.e., for $i=\ell+1,...,n$, the channel erases the bit $x_i$ of $\bx$ whenever $x_i\neq x_i'$ (where $x_i'$ is the $i$th bit of $\bx'$).


We denote the $\ell$-prefix of codeword $\bx$ by $\bx_1$ and $(n-\ell)$-suffix by $\bx_2$. Similarly, the last $n-\ell$ bits of the channel's output sequence $\by$ is denoted by $\by_2$. Before we give the formal statements that constitute the main body of the proof of Theorem~\ref{upper_bound}, we will informally describe the proof steps to give some intuition about the underlying idea of the proof.
First suppose that by the end of the waiting phase the observed prefix of the transmitted codeword is $\bx_1$, then the remaining uncertainty about the message $U$ (at the receiver and the channel) is given by $H(U|\bbx_1=\bx_1)$. Let's consider the set of prefixes $\bx_1$ for which such uncertainty is large enough. Namely, we define
\begin{align}
A_{\epsilon}&\triangleq\{\bx_1: \bx~\text{is a codeword},~ H(U|\bbx_1=\bx_1)>n\frac{\epsilon}{4}\}\label{A_eps}
\end{align}
Our first step of the proof is to show that, for some choice of $\ell$, the probability that the observed prefix $\bbx_1$ of the actual codeword lies in $A_{\epsilon}$ is a strictly positive constant.

Next, \emph{conditioning on such event}, for every prefix $\bx_1\in A_{\epsilon}$, we introduce a list, denoted by $\ml_{\bx_1}$, which contains all the codewords that shares the same prefix $\bx_1$. More formally, $\ml_{\bx_1}$ is defined as
\begin{align}
\ml_{\bx_1}&\triangleq\{\bx'\in\mx^n: \exists u'\in\mU~\text{s.t.}~ \enc(u')=\bx',~\bx_1'=\bx_1\}\label{list_bx_1}
\end{align}
Our goal then is to show that the size of such list is small enough such that a codeword picked up randomly by the channel (according the conditional distribution of $\bbx$ given $\bbx_1=\bx_1$) from such list will, with a strictly positive constant probability, end up being: (i) an encoding of a different message other than the actual message, and (ii) at a Hamming distance less than $pn$ from the actual codeword (corresponding to the actual message) that is originally transmitted. Hence, by pushing the transmission towards such fake codeword in the push phase, the channel will succeed, with a strictly positive constant probability, in fooling the decoder causing it to believe that the transmitted codeword is the fake one and hence rendering a decoding error.

The following lemma constitutes the first step of the proof. In this lemma, we formally give a lower bound on the probability of the event that the $\ell$-prefix $\bbx_1$ of the actual codeword $\bbx=\enc(U)$ lies in the set $A_{\epsilon}$ for a specific choice of $\ell$.

\begin{lemma}\label{step1_ub}
 Let $p\in(0,\frac{1}{2}]$ and $0<\epsilon<4p$. Set $R=R_{\sf{Upper}}(p)+\epsilon$ and $\ell=\left(R_{\sf{Upper}}(p)+\frac{\epsilon}{2}\right)n$. Let $E_1$ denote the event $\{\bbx_1\in A_{\epsilon}\}$ where $A_{\epsilon}$ is given by (\ref{A_eps}). Then, we must have $\pr(E_1)\geq\frac{\epsilon}{4}$
\end{lemma}

The next lemma furnishes the central part of the proof of Theorem~\ref{upper_bound} and highlights the main idea of a successful ``push'' strategy.
\begin{lemma}\label{main_lem_ub}
Let $p,~\epsilon,~\ell,$ and $R$ be as in Lemma~\ref{step1_ub}. Let $\bx_1$ be a legitimate codeword prefix and let $\ml_{\bx_1}$ be as defined in (\ref{list_bx_1}). Let $\bbx'$ denote a codeword that is randomly sampled from $\ml_{\bx_1}$ according to the conditional distribution of the encoder's output $\bbx$ given that $\bbx_1=\bx_1$. Let $U'$ be the message corresponding to $\bbx'$. Let $E_2$ denote the event $\{U'\neq U,~d_{H}\left(\bbx', \bbx\right)\leq np\}$ where $d_{H}(.,.)$ is the Hamming distance between two binary vectors of length $n$, $U$ is the original message, and $\bbx=\enc(U)$ is the original codeword that is being transmitted. Then, we must have
\begin{align}
\pr\left(E_2~\vert~\bbx_1=\bx_1\right)&\geq\epsilon^{O(1/\epsilon)}~ ~\text{whenever }~\bx_1\in A_{\epsilon}\label{lem2_lb}
\end{align}
for all sufficiently large $n$, where $A_{\epsilon}$ is as given by (\ref{A_eps}).
\end{lemma}

We defer the proofs of Lemmas~\ref{step1_ub} and \ref{main_lem_ub} to the appendix. Now, given those lemmas, we are ready to prove Theorem~\ref{upper_bound}.

\paragraph{\bf{Proof of Theorem~\ref{upper_bound}:}}
First, observe that the event $E_2$ of Lemma~\ref{main_lem_ub} guarantees the success of the wait-push strategy since it ensures that the fake codeword $\bbx'$ drawn randomly from $\ml_{\bx_1}$ after observing the prefix $\bx_1$ according to the conditional distribution of $\bbx$ given $\bbx_1=\bx_1$ has the following two properties. The first property is that the fake codeword $\bbx'$ results from the encoding of a different message $U'\neq U$. The second property is that the Hamming distance between the actual codeword $\bbx$ and the fake codeword $\bbx'$ is at most $pn$ and hence the channel would have enough erasures budget to erase those bits of $\bbx$ that differ from the fake codeword $\bbx'$. Thus, the received sequence $\bby$ after the push phase will make the decoder completely uncertain whether the transmitted message is $U$ or $U'$. Therefore, conditioned on $E_2$, a decoding error will occur with probability at least $1/2$.
Thus, to prove our theorem, it suffices to show that $\pr(E_2)$ is strictly positive constant that does not depend on $n$. To do this, we derive a lower bound on $\pr(E_1, E_2)$ where $E_1$ is the event of Lemma~\ref{step1_ub}, namely, the event that the $\ell$-prefix $\bbx_1$ of $\bbx$ lies in the set $A_{\epsilon}$. Observe that, by Lemmas~\ref{step1_ub} and \ref{main_lem_ub}, we have $\pr(E_2)\geq \pr(E_1,E_2)=\sum_{\bx_1\in A_{\epsilon}}\pr\left(E_2~\vert~\bbx_1=\bx_1\right)\pr\left(\bbx_1=\bx_1\right)$ $\geq \epsilon^{O(1/\epsilon)}\pr(E_1)\geq \epsilon^{O(1/\epsilon)}$. This completes the proof of the theorem.

\section{Lower Bound}~\label{LOWER}
In this section, we present a lower bound on $C_p$, denoted by $R_{L}(p)$. As discussed in Section~\ref{system}, we consider the worst-case error scenario where, before transmission starts, an arbitrary message $u\in\mU$ is chosen and is known to the channel (this means that $u$ can actually be chosen by the channel itself before transmission). Our goal is to show, for every $p\in[0,~1/2]$ and every small $\epsilon>0$, the existence of a $(2^{nR},~n)$ code, i.e., an encoder-decoder pair $(\enc,~\dec)$, where $R\geq R_L(p)-\epsilon$, such that for every $\adv\in\ce_p$, the probability of decoding error $P_{max}^n$, defined in (\ref{P_max-e}), can be made arbitrarily small for sufficiently large $n$. 

We will show the existence of such code with a randomized (i.e., stochastic) encoder $\enc$ that takes two inputs. The first input is the message $u\in\mU=\{0,1\}^{nR}$. The second input is a random string $S$ drawn randomly and uniformly from a set $\ms\triangleq\{0,1\}^{\delta n}$ for some small $\delta$ (to be specified later)\footnote{For simplicity of notation, here again we assume, w.l.o.g., that $\delta n$ is an integer.}. The second input to the encoder is generated locally and is not shared with the decoder $\dec$. Hence, we will twist the notation a little bit in this section and write $\enc$ as a two-input function $\enc(u,s),~u\in\mU,~s\in\ms$ to explicitly express the fact that it is a randomized encoder. Our randomized encoder will have a specific form, namely, for an input message $u\in\mU$, the first $Rn$ bits of the encoder's output is the message $u$ itself. Hence, we call it a \emph{systematic} randomized encoder which is formally defined as follows.

\begin{Definition}[Systematic randomized encoder]\label{rand_enc}
Let $R, \delta > 0$, $\mU=\{0,1\}^{Rn}$, and $\ms=\{0,1\}^{\delta n}$. A systematic randomized encoder of a $\left(2^{Rn},~n\right)$ code is a function $\enc:\mU\times\ms\rightarrow\{0,1\}^n$ whose second input is picked uniformly at random from $\ms$. Moreover, for every input $(u,~s)\in\mU\times\ms$, the output of the systematic encoder $\enc(u, s)$ is an $n$-bit codeword $\bx(u,~s)=\left(u,~\bx_2(u, s)\right)\in\mU\times\{0,1\}^{(1-R)n}$.
\end{Definition}

To prove our result, we will use a \emph{random coding argument} in which our systematic randomized encoder $\enc$ is chosen uniformly at random from the class of all systematic randomized encoders. That is, for every $(u,s)\in\mU\times\ms$, the $Rn$-prefix of the output codeword of our encoder is the message $u\in\mU$ while the $(1-R)n$ suffix $\bx_2(u, s)$ is chosen independently and uniformly from $\{0,1\}^{(1-R)n}$. Note that this form of randomness is over the choice of the code is because we adopt a random coding argument as a proof technique and \emph{not} to be confused with the randomness due to the stochastic nature of the encoder, i.e., the randomness due to the uniform choice of $s\in\ms$. That is, we first chose our encoder uniformly at random from the class of encoders satisfying Definition~\ref{rand_enc}, then given a specific choice $\enc$ of our encoder, for every message $u\in\mU$, we sample an $s$ uniformly from $\ms$ and output the codeword $\enc(u,s)$.

The decoder $\dec$ is deterministic function that takes a received vector $\by\in\{0, 1, \wedge\}^n$ as input and returns two outputs: an estimated message $\hat{u}\in\mU$ and also an estimated value for the encoder local randomness $\hat{s}\in\ms$.

We do not have to require that the decoder returns an estimate for the encoder's local randomness since it is not a part of the message. However, by doing so, we actually give a stronger result. We will show that, with high probability, our decoder recovers both the message and the encoder's random coins. Thus, one may think of the message as a pair $(u,~s)$, and the encoder as deterministic. But, correct decoding with high probability is only guaranteed when $s$ is picked uniformly and independently of $u$.

Accordingly, the error criterion, with respect to $(\enc, \dec)$ pair and for a given $\adv\in\ce_p$, is the probability of the worst-case error $P_{max}^n(\enc, \dec, \adv)$ averaged over the set of the encoder's random coins $\ms$. Precisely,
$$P_{max}^n(\enc, \dec, \adv)=\max_{u\in\mU}\frac{1}{2^{\delta n}}\sum_{\by\in\my^n}\pr\left(\adv\left(\enc\left(u,~s\right)\right)=\by\right)\pr\left(\dec(\by)\neq (u,~s)\right).$$
The objective is to show the existence of a $(\enc,~\dec)$ which, for a sufficiently large $n$ and for all $\adv\in\ce_p$, makes $P_{max}^n$ arbitrarily small.

We formally state our lower bound, $R_L(p)$, in the following theorem.
\begin{theorem}\label{lower_bound}
For all $p\in [0,1]$, the rate $R_L(p)$, given below, is achievable for every $\adv\in\ce_p$ and hence $C_p\geq R_L(p)$.
\begin{align}
R_L(p)&=\left\{\hspace{0.5cm}\begin{array}{cc}
          ~1-\frac{p}{\log(4/3)}, & 0\leq p\leq p_1 \\
          r(p), & p_1 < p < 1/2 \\
          0, & 0.5\leq p \leq 1
        \end{array}\right. \nonumber\,,
\end{align}
where $p_1=\frac{3\log(4/3)}{2+3\log(4/3)}\approx 0.384$, the function
$r(p)$ is the unique root $x$ of the equation $G_p(x)=0$ in the interval $[0,~\frac{3}{2}p-\frac{1}{2}]$,
\begin{align}
G_p(x)&=(1-x)H\left(\frac{p-x}{1-x}\right)-1+2x\,,
\end{align}
and $H(.)$ is the binary entropy function.
\end{theorem}

To prove Theorem~\ref{lower_bound}, we will show that our result holds for a setting stronger than the causal setting, namely, the \emph{two-step} model that is analogous to the two-step model for bit flips considered in \cite{HL}, and hence it must hold for $\ce_p$. In the two-step model, the transmission of a codeword occurs in two steps. In the first step, the transmitter sends the first $Rn$ bits of the codeword, i.e., the message $u$. The channel, which already knows these $Rn$ bits since the message is fixed, erases some of those bits. Then, in the second step, the transmitter sends the remaining $(1-R)n$ bits of the codeword, i.e. the suffix $\bx_2$, and the channel, which now sees the whole codeword, erases some of the last $(1-R)n$ bits. The total number of bits the channel can erase in the two steps together is at most $pn$.

The proof relies on the notion of the \emph{forbidden ball}. For a given $n$-bit input $\bx$ to the two-step channel and a given erasure pattern chosen by the channel in \emph{the first step}, the forbidden ball is a subset of $\{0,1\}^n$ that contains every $n$-bit string $\bx'$ for which there is a legitimate erasure pattern that the channel can choose in the second step such that, upon observing the whole output of the channel, $\bx'$ and $\bx$ will be equally likely to be the input to the channel. To clarify, suppose that $\bx=(u,~\bx_2)\in\mU\times\{0,1\}^{(1-R)n}$ is the input to the channel in the two-step model. By the end of the first step, the channel decides to erase, say, $qn$ bits in the prefix $u$ where $q\leq \min(p, R)$ resulting in a vector $\by_1\in\{0,1,\wedge\}^{Rn}$. We say that a vector $\bx'\in\{0,1\}^m$ is \emph{consistent with} $\by\in\{0,1,\wedge\}^m$ if $x'_i=y_i$ for all $i$ such that $y_i\neq\wedge$, where $x'_i$ (resp., $y_i$) is the $i$th bit of $\bx'$ (resp., $\by$), that is, if $\bx'$ and $\by$ agree in every non-erased entry. We denote the number of erasures in a vector $\by$ by $\sharp(\by)$. Now, consider the intermediate $n$-bit vector $(\by_1,~\bx_2)$ right after the action of the channel in the first step and before the second step. For $q=\frac{\sharp(\by_1)}{n}$, the \emph{forbidden ball} $B_{R}^{p,q}(\by_1,\bx_2)$ centered at $(\by_1,~\bx_2)$ defined as
\begin{align}
&B_{R}^{p,q}(\by_1,\bx_2)=\big\{\bx'_1\in\{0,1\}^{Rn}:~\bx'_1~\text{ consistent with }\by_1\big\}\times\big\{\bx'_2\in\{0,1\}^{(1-R)n}:~d_H(\bx'_2,\bx_2)\leq(p-q)n\big\}\nonumber
\end{align}
Note that the forbidden ball $B_{R}^{p,q}(\by_1,~\bx_2)$ is the product set of a Hamming cube with a Hamming ball. This set contains all vectors $\bx'\in\{0,1\}^n$ such that it is possible for the channel to erase bits in the second step (given that it already erased $qn$ bits in the first step resulting in $\by_1$) to make the receiver believe that $\bx'$ was a possible input string to the channel. 

Clearly, the original input $\bx=(\bx_1,\bx_2)$ lies in $B_{R}^{p,q}(\by_1,~\bx_2)$. If $\bx$ is a codeword and is the only codeword lying in the forbidden ball, then, in this case, the decoder can recover the original message successfully with no error. Our goal, roughly speaking, is to show the existence of a code where this holds for ``most'' of the codewords. Using a random coding argument, one can show that, roughly speaking, a ``good'' $(2^{Rn}, n)$ code that achieves a rate $R$ in the two-step model exists when the size of any such forbidden ball is smaller than $2^{(1-R)n}$ by an exponential factor. To do this, a crucial step in the existence proof is to characterize the size of such a ball. Note that the size of $B_{R}^{p,q}(\by_1,~\bx_2)$ does not depend on $(\by_1,~\bx_2)$. Hence, we will use $B_{R}^{p,q}$ to denote the size of $B_{R}^{p,q}(\by_1,~\bx_2)$. Lemma~\ref{forbid_ball_size} below gives an upper bound on $B_{R}^{p,q}$ for any $p\in[0,~1/2]$ and $q\in[0,~\min(R, p)]$ when $R$ is carefully chosen. Before stating this lemma, we first give the following definition.
\begin{Definition}\label{def_R_delta_eta}
For every $p\in(0, 1/2)$ and every $\delta,~\eta>0$, define $R_{\delta, \eta}(p)$ as
\begin{align}
&\hspace{-.05cm}R_{\delta, \eta}(p)\hspace{-.1cm}=\hspace{-.1cm}\left\{\hspace{-.1cm}\begin{array}{cc}
          \hspace{-.08cm}1-\frac{p}{\log(4/3)}-\frac{1-\log(4/3)}{\log(4/3)}\delta-\frac{\eta}{\log(4/3)}, & \hspace{-0.12cm}0< p < p_1 \\
          \hspace{-.25cm}r_{\delta, \eta}(p) +\delta, & \hspace{-.18cm}p_1 \leq p < 1/2
        \end{array}\right.\label{R_delta_eta}
\end{align}
where $p_1=\frac{3\log(4/3)}{2+3\log(4/3)}\approx 0.384$ and $r_{\delta, \eta}(p)$ is the unique solution of the equation $G_p(x)+\delta+\eta=0, ~x\in [0,~\frac{3}{2}p-\frac{1}{2}]$ (for $x$) where
\begin{align}
G_p(x)&=(1-x)H(\frac{p-x}{1-x})-1+2x.\label{S_p_of_x}
\end{align}
\end{Definition}

\begin{lemma}\label{forbid_ball_size}
Let $p\in (0,~1/2)$. For all sufficiently small $\delta,~\eta>0$, for
all sufficiently large $n$, if $R = R_{\delta, \eta}(p)-\delta$ (where $R_{\delta, \eta}(p)$ is as in Definition~\ref{def_R_delta_eta}), then for every $q\in[0,~\min(R, p)]$, the size $B_{R}^{p,q}$ of the forbidden ball is bounded as
\begin{align}
B^{p, q}_{R}&\leq 2^{\left(1-R_{\delta, \eta}(p)-\eta/2\right)n}\,.
\end{align}
\end{lemma}

The following claim will be used later to complete the proof of our main result.
\begin{claim}\label{limit_rate}
For every $p\in(0, 1/2)$, the quantity $R_{\delta, \eta}(p)$, defined in Definition~\ref{def_R_delta_eta}, satisfies
\begin{align}
\lim_{\delta+\eta\rightarrow 0}R_{\delta, \eta}(p)&=R_L(p)
\end{align}
where $R_L(p)$ is as given by Theorem~\ref{lower_bound}.
\end{claim}

To prove Theorem~\ref{lower_bound}, we consider two ways in which a decoding error can occur in the two-step model. The first is when the decoder decides that the true codeword is one whose $R n$-prefix is different from that of the originally transmitted codeword. This is tantamount to having an erroneous estimate for the message $u\in\mU$ at the decoder's output since the first $Rn$ bits of our encoder's output is the message $u$. The second type of a decoding error is when the decoder believes that the true codeword is one that shares the same $R n$-prefix as the originally transmitted codeword (hence resulting in a correct estimate for the message $u\in\mU$) but a different suffix from the original codeword (hence resulting in the wrong estimate for $s\in\ms$). We refer to the former type of errors as \emph{type-I errors} while we refer to the later as \emph{type-II errors}. Note that, if we are not interested in estimating $s\in\ms$, then type-II errors are irrelevant and we can definitely ignore them. However, pursuing a stronger result, we show the existence of a code of rate $R_L(p)$ that is capable of correcting both types of errors with probability approaching $1$ as $n\rightarrow\infty$.


Let $S$ be a random variable that is uniformly distributed over $\ms$ and let $u\in\mU$. Suppose that $(u,~S)$ is the (message, random coins) pair. Let $\left(u, \bx_2(u,S)\right)\in\{0,1\}^n$ be the transmitted codeword. In the first step, the channel erases $qn$ bits of the prefix $u$ resulting in the vector $\by_1\in\{0,1,\wedge\}^{R n}$ for some $q\in[0,~\min(R,p)]$. In the second step, the channel erases at most $(p-q)n$ bits of $\bx_2(u, S)$. An error occurs, either of type-I or type-II, when the forbidden ball $B_{R}^{p,q}\left(\by_1, \bx_2(u, S)\right)$ contains at least one legitimate codeword $(u', \bx_2')$ other than $\left(u, \bx_2(u,~S)\right)$.

\subsection{Type-I Errors:}\label{typeI}
Here, we study the case where there is at least one other codeword $(u', \bx_2')$ in $B_{R}^{p,q}\left(\by_1, \bx_2(u, S)\right)$ such that $u'\neq u$. In other words, a type-I error occurs, with respect to $(u,~\by_1)$, whenever there exists a message $u'\neq u$ and some $s'\in\ms$ whose corresponding codeword $\left(u', \bx_2(u',s')\right)$ lies inside the forbidden ball $B_{R}^{p,q}\left(\by_1, \bx_2(u, S)\right)$. Let us denote this event by $\ei$. Formally, the event $\ei$ is defined as
\begin{align}
\ei&\triangleq \big\{\exists (u',s')\in\left(\mU\setminus\{u\}\right)\times\ms:\enc(u',s')\in B_{R}^{p,q}\left(\by_1,~\bx_2(u, S)\right)\big\}\label{ei}
\end{align}

In Definition~\ref{good_code_def} below, we define a property that, if possessed by a systematic code, would lead to a vanishing probability of type-I errors.

%

\begin{Definition}[$\tilde{\eta}$-good systematic code w.r.t. $(u, \by_1)$]\label{good_code_def}
Fix $R,~\delta > 0$ and let $\mU=\{0,1\}^{Rn},$ and $\ms=\{0,1\}^{\delta n}$. Let $p\in[0,~1/2]$ and fix some $q\in[0,~\min(R, p)]$. Let $S$ be a random string that is uniformly distributed over $\ms$. Fix $u\in\mU$ and let $\by_1\in\{0,1,\wedge\}^{R n}$ be the resulting vector after erasing some $qn$ bits of $u$. Let $\tilde{\eta}>0$. A $\left(2^{Rn}, n\right)$ code is a $\tilde{\eta}$-good systematic code with respect to $(u,~\by_1)$ if it is associated with a systematic encoder $\enc$ (as defined in (\ref{rand_enc})) such that
\begin{align}
\peI\triangleq\pr\left(\ei\right)&\leq 2^{-\tilde{\eta}n}\nonumber
\end{align}
where $\ei$ is as defined in (\ref{ei}) and the probability is over the choice of $S$.
\end{Definition}

The next lemma shows that, in a two-step model with erasure rate $p\in[0, 1/2]$, a systematic code, with rate arbitrarily close to our claimed lower bound $R_{L}(p)$, whose codewords' suffixes are chosen uniformly at random is $\tilde{\eta}$-good (for some fixed $\tilde{\eta}>0$) \emph{with respect to all pairs} $(u, \by_1)$ with overwhelming probability for sufficiently large $n$. This asserts the existence of at least one code for the two-step model which is $\tilde{\eta}$-good with respect to all pairs $(u, \by_1)$ which implies the existence of a code of rate arbitrarily close to $R_L(p)$ that can correct all type-I errors with probability arbitrarily close to $1$ for sufficiently large $n$.

\begin{lemma}\label{good_rand_code}
Let $p\in[0,~1/2]$. For sufficiently small $\delta, \eta >0$, let $R=R_{\delta, \eta}(p)-\delta$ where $R_{\delta, \eta}(p)$ is as in Definition~\ref{def_R_delta_eta}. Let $\mU=\{0,1\}^{Rn}$ and $\ms=\{0,1\}^{\delta n}$. Let $\enc:\mU\times\ms\rightarrow\{0,1\}^n$ be a systematic randomized encoder (as defined in Definition~\ref{rand_enc}) such that, for every $u\in\mU$, $s\in\ms$, $\bx_2(u, s)$ is chosen independently and uniformly from $\{0,1\}^{(1-R)n}$. With probability at least $1-e^{-2^{\Omega(n)}}$ over the choice of $\enc$, the code associated with $\enc$ is $\frac{\eta}{4}$-good with respect to all pairs $(u, \by_1)\in\mU\times\{0,1,\wedge\}^{R n}$ where $\by_1$ is the resulting vector after erasing at most $\min(R, p)n$ bits of $u$.
\end{lemma}

\subsection{Type-II Errors:}\label{typeII}

Here, we consider the case where the decoder outputs the correct $u\in\mU$ but the wrong $s\in\ms$. In other words, when $(u, s)\in\mU\times\ms$ is the (message, random coins) pair, we consider the error event that occurs when the decoder confuses the actual codeword $\left(u, \bx_2(u,s)\right)$ with some other codeword $(u, \bx_2')$ for some $\bx_2'\neq\bx_2(u,s)$. Our goal is to show that, when our systematic encoder $\enc$ is such that, for every $u\in\mU$ and $s\in\ms$, $\bx_2(u, s)$ is chosen uniformly at random from $\{0,1\}^{(1-R)n}$, then with an overwhelming probability over the choice of $\enc$, such error event does not occur for ``almost'' all $(u,s)\in\mU\times\ms$.

Roughly speaking, we need to show that $\enc$ has the property that, except for a few pairs of codewords, any pair of codewords that share the same prefix (i.e., that correspond to the same message $u\in\mU$) are not very close to each other in the Hamming distance. More precisely, for $p\in[0, 1/2)$, except for a small subset of codewords, any pair of codewords that share the same prefix will, with high probability (over the choice of the code), be at Hamming distance greater than $pn$. In fact, the following lemma gives us what we are looking for. The following lemma is closely related to Lemma~III.4 in \cite{HL}. The proof of the following lemma follows from the standard distance argument in the proof of the GV bound. The proof is omitted since it follows similar steps to that of Lemma~III.4 in \cite{HL}.

\begin{lemma}\label{errors_for_same_pefix_lem}
Let $p\in[0, 1/2)$, $0< R < 1-2p$ and let $\delta>0$ be sufficiently small. Let $\enc:\mU\times\ms\rightarrow\{0,1\}^n$ be a systematic encoder (as in Definition~\ref{rand_enc}) such that, for every $u\in\mU$ and $s\in\ms$, $\bx_2(u, s)$ is chosen independently and uniformly from $\{0,1\}^{(1-R)n}$. There exists a $\gamma>0$ for which the following holds for all sufficiently large $n$. With probability at least $1-e^{-2^{\Omega(n)}}$ over the choice of the encoder $\enc$, a code associated with $\enc$ satisfies the following: There exists a set $\mathcal{V}\subset\mU$ with $\vert\mathcal{V}\vert\leq 2^{(R-\gamma)n}$ such that for every $u\in\mU\setminus\mathcal{V}$, there exists $\mathcal{Q}_{u}\subseteq\ms$ of size $\vert\mathcal{Q}_{u}\vert<2^{(\delta-\gamma)n}$ such that for every distinct $s, s'\in\ms\setminus\mathcal{Q}_{u},$ we have $d_H\left(\bx_2(u,s), \bx_2(u,s')\right)> pn$.
\end{lemma}
\subsection{Proof of Theorem~\ref{lower_bound}:}
Let $p\in(0,~1/2)$. Let $\delta,~\eta>0$ be sufficiently small, $R_{\delta, \eta}(p)$ be as in Definition~\ref{def_R_delta_eta}, $R=R_{\delta, \eta}-\delta$, $\mU=\{0,1\}^{Rn}$, and $\ms=\{0,1\}^{\delta n}$. Let $\enc$ be a systematic randomized encoder that satisfies the conditions in Lemmas~\ref{good_rand_code} and \ref{errors_for_same_pefix_lem} simultaneously. Note that the existence of such encoder is guaranteed since, by Lemmas~\ref{good_rand_code} and \ref{errors_for_same_pefix_lem}, the probability that a systematic randomized encoder $\enc$ satisfies the conditions of those lemmas simultaneously is at least $1-e^{-2^{\Omega(n)}}$. Let $\gamma>0$, $\mathcal{V}\subset\mU$, and $\{\mathcal{Q}_{u},~u\in\overline{\mathcal{V}}\}$ be as in Lemma~\ref{errors_for_same_pefix_lem} where $\overline{\mathcal{V}}=\mU\setminus\mathcal{V}$.

Let $\overline{\mathcal{Q}}_u=\ms\setminus\mathcal{Q}_u,~u\in\overline{\mathcal{V}}$. Let's order the members of $\overline{\mathcal{Q}}_u$, say, lexicographically, and denote them by $s_{\overline{\mathcal{Q}}_u}\left(1\right)< ... <s_{\overline{\mathcal{Q}}_u}\left(\vert\overline{\mathcal{Q}}_u\vert\right)$. Let $T$ denote $\min_{u\in\overline{\mathcal{V}}}\vert\overline{\mathcal{Q}}_u\vert$. Hence, by Lemma~\ref{errors_for_same_pefix_lem}, $T\geq 2^{\delta n}-2^{(\delta-\gamma)n}\geq2^{\delta n-1}$ for sufficiently large $n$. Define $\overline{\mt}\triangleq\{1,...,2^{\delta n-1}\}$.
Define a new encoder $\widetilde{\enc}:\overline{\mathcal{V}}\times\overline{\mt}$ as follows. For every $(u,t)\in\overline{\mathcal{V}}\times\overline{\mt}$, $\widetilde{\enc}(u,~t)=\enc(u,~s_{\overline{\mathcal{Q}}_u}(t))$.
Note that, by Lemma~\ref{errors_for_same_pefix_lem}, $\vert\overline{\mathcal{V}}\vert\geq 2^{Rn}-2^{(R-\gamma)n}\geq 2^{Rn-1}$ for sufficiently large $n$. This, together with the fact that $\vert\overline{\mt}\vert=2^{\delta n-1}$ (as shown above), implies that $\widetilde{\enc}$ has the same asymptotic rate as $\enc$. Namely, $\lim_{n\rightarrow\infty}\log\left(\vert\overline{\mathcal{V}}\vert\cdot\vert\overline{\mt}\vert\right)=R+\delta=R_{\delta,\eta}(p)$.

Thus, it remains to show that the probability of decoding error $P_{max}^n$ with respect to $\widetilde{\enc}$ decays to zero as $n\rightarrow\infty$. To do this, we consider each of the two types of decoding error. 
First, for type-I errors, since the code associated with $\enc$ is a $\frac{\eta}{4}$-good systematic code for all $(u,~\by_1)$ whose existence is shown by Lemma~\ref{good_rand_code}, then the probability of such type of errors is bounded from above by $\frac{2^{(\delta-\frac{\eta}{4})n}}{2^{\delta n-1}}=2^{-\frac{\eta}{4}n+1}$. For errors of type-II, since $u\in\overline{\mathcal{V}}$ and since we pruned the code associated with $\enc$ by getting rid off the all the ``bad'' $s$, namely, those in $\bigcup_{u\in\overline{\mathcal{V}}}\mathcal{Q}_{u}$. Hence, according to Lemma~\ref{errors_for_same_pefix_lem}, any pair of codewords that have the same prefix are at Hamming distance strictly greater than $pn$. Since the channel can only erase at most $pn$ bits in the second step, the decoder will always know which codeword was originally transmitted and hence decodes successfully. Thus, the pruned code associated with $\widetilde{\enc}$, can correct all type-II errors .

Summing up, the probability of decoding error according to the above analysis of the two types of errors is $\pr(\text{Type-I error occurs})\leq 2^{-\frac{\eta}{4}n+1}=2^{-\Omega(n)}$ which can be made arbitrarily small for sufficiently large $n$.

\appendix

\section{Proofs of Section~\ref{UPPER}}
We give here the proofs of Lemmas~\ref{step1_ub} and \ref{main_lem_ub} of Section~\ref{UPPER}.
\vspace{0.5cm}
\subsection{Proof of Lemma~\ref{step1_ub}:}~
Note that we have $I(U;\bbx_1)\leq H(\bbx_1)\leq\ell=(R_{\sf{Upper}}(p)+\frac{\epsilon}{2})n$. Thus,
\begin{align}
H(U|\bbx_1)&\geq H(U)-(R_{\sf{Upper}}(p)+\frac{\epsilon}{2})n=nR-(R_{\sf{Upper}}(p)+\frac{\epsilon}{2})n=n\frac{\epsilon}{2}\nonumber
\end{align}
By Markov's inequality,
\begin{align}
\pr\left(nR-H(U|\bbx_1=\bx_1)\geq nR-n\frac{\epsilon}{4}\right)&\leq\frac{nR-n\frac{\epsilon}{2}}{nR-n\frac{\epsilon}{4}}\nonumber\\
&=1-\frac{\frac{\epsilon}{4}}{R-\frac{\epsilon}{4}}\nonumber
\end{align}
Since $\epsilon\leq R\leq 1$, it follows that $\pr\left(H(U|\bbx_1=\bx_1)\geq n\frac{\epsilon}{4}\right)\geq\frac{\epsilon}{4}$

\subsection{Proof of Lemma~\ref{main_lem_ub}:}
First, we give the following two lemmas that will be useful in proving Lemma~\ref{main_lem_ub}.

\begin{lemma}\label{plotkin}(\textbf{Plotkin's Bound}\footnote{Here, we give the version of Plotkin's bound on the average distance of a binary code.}~\cite{PLTKN})
A binary $(M,n)$ code $\mc$ must satisfy $M\leq\frac{2d_{avg}}{2d_{avg}-n}$ whenever the average distance $d_{avg}> \frac{n}{2}$ where $d_{avg}=\frac{1}{M(M-1)}\sum_{\bx,\by\in\mc}d_H(\bx,\by)$.
\end{lemma}

\begin{lemma}\label{suff_entropy}(\cite{DJLS}, Lemma~3)
Let $V$ be a random variable on a discrete finite set $\mv$ with entropy $H(V)\geq\lambda$, and let $V_1, V_2,..., V_m$ be i.i.d. copies of $V$. Then
\begin{align}
&\hspace{-.05cm}\pr\left(\{V_i: i=1,...,m\} \text{are\hspace{-.05cm} all distinct}\right)\hspace{-.05cm}\geq\hspace{-.1cm}\left(\frac{\lambda-1-\log(m)}{\log(\vert\mv\vert)}\right)^{m-1}.\nonumber
\end{align}
\end{lemma}

Now, consider the situation by the end of the ``wait'' phase. Let $\bbx_1=\bx_1$ for some $\bx_1\in A_{\epsilon}$. Suppose that we sample (with replacement) $m$ codewords $\bbx(1), ..., \bbx(m)$ from $\ml_{\bx_1}$ according to the conditional distribution of the encoder's output given the observed $\ell$-prefix $\bx_1$ denoted by $\pr_{\bbx\vert\bbx_1=\bx_1}$. That is, $\bbx(1), ..., \bbx(m)$ are i.i.d. with distribution $\pr_{\bbx\vert\bbx_1=\bx_1}$. Let $U_1,...,U_m$ denote the corresponding messages, respectively. Let $\bbx'$ denote the codeword drawn from $\ml_{\bx_1}$ by the adversary by the end of the ``wait'' phase and $U'$ denote the corresponding message. Note that the transmitted codeword $\bbx$ and the adversary's codeword $\bbx'$ are independent and identically distributed according to $\pr_{\bbx\vert\bbx_1=\bx_1}$ in the same fashion any pair in the set of $m$ codewords $\{\bbx(1), ..., \bbx(m)\}$ mentioned above are independent and identically distributed. Similarly, the original message $U$ and the adversary's message $U'$ are independent and identically distributed according to $\pr_{U\vert\bbx_1=\bx_1}$ (the conditional distribution of the message $U$ given that $\bbx_1=\bx_1$) in the same fashion any pair in the set of $m$ messages $\{U_1, ..., U_m\}$ mentioned above are independent and identically distributed.

\begin{prop}\label{distinctness_of_m_messages}
Let $E_3$ denote the event $\{U_1, ..., U_m$ are all distinct$\}$ for some integer $m$. Then, for sufficiently large $n$, we have $\pr(E_3~\vert~\bbx_1=\bx_1)\geq(\frac{\epsilon}{5})^{m-1}$.
\end{prop}
\begin{proof}
The proof follows from the fact that $\bx_1\in A_{\epsilon}$ and the result of Lemma~\ref{suff_entropy} above. We apply Lemma~\ref{suff_entropy} with the distribution of $V$ set to the distribution of $U$ conditioned on $\bbx_1=\bx_1$ and $\lambda=n\frac{\epsilon}{4}$. Hence, we get
\begin{align}
\pr(E_3~\vert~\bbx_1=\bx_1)&\geq\left(\frac{n\epsilon/4-\log(m)-1}{n}\right)^{m-1}\nonumber\\
&>\left(\frac{\epsilon}{4}-\frac{\log(m)}{n}\right)^{m-1}\nonumber
\end{align}
For fixed $m$ and sufficiently large $n$, we have $\log(m)\leq\frac{\epsilon}{20}n$. Hence, the proof is complete.
\end{proof}

 Let $\mc_m$ denote the collection of codewords $\{\bbx(1), ..., \bbx(m)\}$ picked in the fashion described above for some $m$ (to be decided later). Let $d_{avg}(\mc_m)$ be the average Hamming distance between any pair in $\mc_m$ defined as
\begin{align}
d_{avg}(\mc_m)&\triangleq\frac{1}{m(m-1)}\sum_{i\neq j}d_{H}(\bbx(i),\bbx(j))\nonumber
\end{align}
where the sum is over all distinct $i,~j$ in $\{1, ..., m\}$. Note that $d_{avg}(\mc_m)$ is a random variable. Now, suppose we condition on both events $\{\bbx_1=\bx_1\}$ and $E_3$ (Note that the former event has been already conditioned upon from the beginning of the proof). The following upper bound on $d_{avg}(\mc_m)$ holds with probability $1$.
\begin{align}
d_{avg}(\mc_m)&\leq\frac{1}{2}\frac{m}{m-1}(n-\ell)=\frac{m}{m-1}(p-\epsilon/4)n\label{plotkin1}
\end{align}
This bound follows directly from Plotkin's bound stated in Lemma~\ref{plotkin} above and the choice $\ell$ specified in the lemma statement.
By setting $m=\frac{9}{\epsilon}$, we further upper bound the right-hand side of (\ref{plotkin1}) to get
\begin{align}
d_{avg}(\mc_m)&\leq np - n\frac{\epsilon}{8}\nonumber
\end{align}
Hence, conditioned on $\{\bbx_1=\bx_1,~E_3\}$, the expected average Hamming distance $E\left[d_{avg}(\mc_m)\vert~\bbx_1=\bx_1,~E_3\right]$ is upper bounded as
\begin{align}
E\left[d_{avg}(\mc_m)\vert~\bbx_1=\bx_1,~E_3\right]&\leq np - n\frac{\epsilon}{8}\label{plotkin3}
\end{align}
On the other hand, we have
\begin{align}
E\left[d_{avg}(\mc_m)\vert~\bbx_1=\bx_1,~E_3\right]&=\frac{1}{m(m-1)}\sum_{i\neq j}E\left[d_{H}(\bbx(i),\bbx(j))\vert~\bbx_1=\bx_1,~E_3\right]\nonumber\\
&=E\left[d_{H}(\bbx(1),\bbx(2))\vert~\bbx_1=\bx_1,~E_3\right]\label{expected_d_avg_Cm2}
\end{align}
where (\ref{expected_d_avg_Cm2}) is due to symmetry, i.e., the fact that, conditioned on $\{\bbx_1=\bx_1,~E_3\}$, the distributions of all pairs $(\bbx(i),~\bbx(j))$ for $i\neq j$ in $\{1, ..., m\}$ are identical.
From (\ref{plotkin3}) and (\ref{expected_d_avg_Cm2}), we get
\begin{align}
E\left[d_{H}(\bbx(1),\bbx(2))\vert~\bbx_1=\bx_1,~E_3\right]&\leq np - n\frac{\epsilon}{8}\nonumber
\end{align}
Thus, by Markov's inequality, we have
\begin{align}
\pr\left(d_{H}(\bbx(1),\bbx(2))> np~\vert~\bbx_1=\bx_1, E_3\right)&\leq\frac{E\left[d_{H}(\bbx(1),\bbx(2))~\vert~\bbx_1=\bx_1, E_3\right]}{np}\nonumber\\
&\leq \frac{np - n\frac{\epsilon}{8}}{np}=1-\frac{\epsilon}{8p}\label{markov_d_avg1}
\end{align}

Finally, we derive bound (\ref{lem2_lb}) in the lemma statement as follows.
\begin{align}
\pr\left(d_{H}(\bbx,\bbx')> np, U\neq U'~\vert~\bbx_1=\bx_1\right)&=\pr\left(d_{H}(\bbx(1),\bbx(2))> np, U_1\neq U_2~\vert~\bbx_1=\bx_1\right)\label{lemma2_lb_fin1}\\
&\geq\pr\left(d_{H}(\bbx(1),\bbx(2))> np, E_3~\vert~\bbx_1=\bx_1\right)\label{lemma2_lb_fin2}\\
&=\pr\left(d_{H}(\bbx(1),\bbx(2))> np~\vert~\bbx_1=\bx_1,~E_3\right)\pr(E_3\vert~\bbx_1=\bx_1)\nonumber\\
&\geq \frac{\epsilon}{8p}(\frac{\epsilon}{5})^{9/\epsilon-1}= \epsilon^{O(1/\epsilon)}\label{lemma2_lb_fin4}
\end{align}
where (\ref{lemma2_lb_fin1}) follows from the fact that the joint distribution of $\left(\bbx, \bbx', U, U'\right)$ is the same as that of $\left(\bbx(1), \bbx(2), U_1, U_2\right)$, (\ref{lemma2_lb_fin2}) follows from the fact that $E_3$ implies $U_1\neq U_2$, and (\ref{lemma2_lb_fin4}) follows from Proposition~\ref{distinctness_of_m_messages} and (\ref{markov_d_avg1}). This completes the proof of Lemma~\ref{main_lem_ub}.

\vspace{0.5cm}
\section{Proofs of Section~\ref{LOWER}}
We give here the proofs of Lemma~\ref{forbid_ball_size}, Claim~\ref{limit_rate}, and Lemma~\ref{good_rand_code} of Section~\ref{LOWER}.
\vspace{0.5cm}
\subsection{Proof of Lemma~\ref{forbid_ball_size}:~}
For now, suppose $p\in(0,1/2),~1-2p\leq R<1$ and $0\leq q\leq\min(p, R)$. We start by giving a general upper bound for $B^{p, q}_{R}$ using the standard bounds on Hamming balls. First, it is easy to see that $B^{p, q}_{R}$ is given by
\begin{align}
B^{p, q}_{R}&=2^{qn} \sum_{i=0}^{(p-q)n} {{(1-R)n}\choose {i}}\nonumber
\end{align}
By using the standard bound on the Hamming ball of radius $(p-q)n$ in $\{1,0\}^{(1-R)n}$ whose exact size is $\sum_{i=0}^{(p-q)n}{{(1-R)n}\choose {i}}$, we get
\begin{align}
B^{p, q}_{R}&\leq 2^{qn}\cdot2^{\left((1-R)H(\frac{p-q}{(1-R)})+\beta\right) n}\label{general_bound_on_B_alpha}
\end{align}
for some arbitrarily small $\beta>0$ for sufficiently large $n$. Note that such bound is valid since $\frac{p-q}{1-R}\leq 1/2$.

Now, we make specific choices for our parameters. Let $p\in(0,~1/2)$. Let $\delta,~\eta>0$ be chosen such that
\begin{align}
\hspace{-.05cm}\delta + \eta &~\leq~ \left\{ \begin{array}{cc}
                      \hspace{-.25cm}3\log(4/3) -1 & ~~\text{if}~0<p<1/3 \\
                      \hspace{-.1cm}\frac{3}{2}\log(4/3)-(\frac{3}{2}\log(4/3)+1)p &  ~ \text{ if}~1/3\leq p<p_1 \\
                      \hspace{-.25cm}1-H(p) &  ~\text{ if } p_1\leq p< 1/2
                    \end{array}\right.\label{bound_on_delta_eta}
\end{align}
where $p_1=\frac{3\log(4/3)}{2+3\log(4/3)}$. Let $R_{\delta, \eta}(p)$ be as given in the lemma statement and $R=R_{\delta, \eta}(p)-\delta$. Hence, for $q\in[0,~\min\left(R_{\delta, \eta}(p)-\delta, p\right)]$, (\ref{general_bound_on_B_alpha}) can be written as
\begin{align}
B^{p, q}_{R}&\leq 2^{(f_{p, \delta, R_{\delta, \eta}(p)}(q)+\beta)n}\label{gen_ball_bound}
\end{align}
where
\begin{align}
\hspace{-.1cm}f_{p, \delta, R_{\delta, \eta}(p)}(q)&=q\hspace{-.07cm}+\hspace{-.05cm}(1\hspace{-.05cm}-\hspace{-.05cm}R_{\delta, \eta}(p)\hspace{-.05cm}+\hspace{-.05cm}\delta)H\left(\frac{p-q}{1-R_{\delta, \eta}(p)+\delta}\right)\nonumber
\end{align}

Now, let's consider the simple optimization problem where we seek to maximize $f_{p, \delta, R_{\delta, \eta}(p)}(q)$ over $q\in[0, \min\left(R_{\delta, \eta}(p)-\delta, p\right)]$. First, fix some $p\in(0,~p_1)$. It is not difficult to see that the maximizer $q^{\ast}$ of $f_{p, \delta, R_{\delta, \eta}(p)}(q)$ in $[0,~\min\left(R_{\delta, \eta}(p)-\delta, p\right)]$ is given by
\begin{align}
&q^{\ast}=p-\frac{1}{3}(1-R_{\delta, \eta}(p)+\delta) ~\text{ whenever } ~0\leq p-\frac{1}{3}(1-R_{\delta, \eta}(p)+\delta)\leq R_{\delta, \eta}(p)-\delta\label{q_star_first_range}
\end{align}
The first two constraints on $\delta+\eta$ in (\ref{bound_on_delta_eta}) and the setting of $R_{\delta, \eta}(p)$ for $p\in(0,~p_1)$ in (\ref{R_delta_eta}) imply the condition in (\ref{q_star_first_range}). Thus, one can easily verify that, for every $p\in(0, p_1)$, we have
\begin{align}
f_{p, \delta, R_{\delta, \eta}(p)}(q)&\leq f_{p, \delta, R_{\delta, \eta}(p)}(q^{\ast})=1-R_{\delta, \eta}(p)-\eta~~\quad \forall~q\in[0,~\min(R, p)]\label{ball_bound_1}
\end{align}

Next, fix some $p\in[p_1, 1/2)$. One can easily verify that the maximizer $q^{\ast}$ of $f_{p, \delta, R_{\delta, \eta}(p)}(q)$ over $[0,~\min(R, p)]$ (note that $R = R_{\delta, \eta}(p)-\delta$) is given by
\begin{align}
&q^{\ast}=R_{\delta, \eta}(p)-\delta ~\text{  whenever }~ R_{\delta, \eta}(p)-\delta < p-\frac{1}{3}(1-R_{\delta, \eta}(p)+\delta)\label{q_star_second_range}
\end{align}
Now, we show that the last constraint on $\delta+\eta$ in (\ref{bound_on_delta_eta}) and the setting of $R_{\delta, \eta}(p)$ for $p\in[p_1,~1/2)$ in (\ref{R_delta_eta}) imply the condition in (\ref{q_star_second_range}). Observe that $r_{\delta, \eta}(p)$ in (\ref{R_delta_eta}), if exists, is the root of $G_p(x)+\delta+\eta=0$, where $G_p(x)$ is given by (\ref{S_p_of_x}), in the interval $[0, \frac{3}{2}p-\frac{1}{2})$. Hence, we have $0\leq R_{\delta, \eta}(p)-\delta<\frac{3}{2}p-\frac{1}{2}$ which implies the condition in (\ref{q_star_second_range}). Thus, it is left to show that there exists a root of $G_p(x)+\delta+\eta=0$ in $[0, \frac{3}{2}p-\frac{1}{2})$. To do this, notice that $G_p(x)$ is strictly increasing in $x$ over the interval $[0, \frac{3}{2}p-\frac{1}{2})$ with $G_p(0)=-(1-H(p))$ and $G_p(\frac{3}{2}p-\frac{1}{2})\geq 0~\forall~p\in[p_1,~\frac{1}{2})$. This together with the fact that $0<\delta+\eta\leq 1-H(p)$ (last constraint on $\delta+\eta$) implies the existence of a unique root $r_{\delta, \eta}(p)$ of $G_p(x)+\delta+\eta$ in the interval $[0, \frac{3}{2}p-\frac{1}{2})$.

Thus, for every $p\in[p_1,~1/2)$, we have
\begin{align}
f_{p, \delta, R_{\delta, \eta}(p)}(q)&\leq f_{p, \delta, R_{\delta, \eta}(p)}(q^{\ast})\nonumber\\
&=r_{\delta, \eta}(p)+(1-r_{\delta, \eta}(p))H(\frac{p-r_{\delta, \eta}(p)}{1-r_{\delta, \eta}(p)})\nonumber\\
&=1-r_{\delta, \eta}(p)-\delta-\eta\label{ball_bound_2a}\\
&=1-R_{\delta, \eta}(p)-\eta~~ \quad\forall~q\in[0,~\min(R, p)]\label{ball_bound_2b}
\end{align}
where  (\ref{ball_bound_2a}) follows from the fact that $r_{\delta, \eta}(p)$ is the root of $G_p(x)+\delta+\eta=0$.

Therefore, for every $p\in(0, 1/2)$, from (\ref{gen_ball_bound}), (\ref{ball_bound_1}), and (\ref{ball_bound_2b}), we get
\begin{align}
B^{p, q}_{R}&\leq 2^{\left(1-R_{\delta, \eta}(p)-\eta+\beta\right)n}~~ \quad\forall~q\in[0,~\min(R, p)]\nonumber
\end{align}
By choosing $n$ to be sufficiently large, we can make $\beta<\frac{\eta}{2}$ and thus we get the desired upper bound in the lemma. This completes the proof.
\vspace{0.5cm}

\subsection{Proof of Claim \ref{limit_rate}:~} We note that, for \mbox{$p\in(0,~p_1)$}, the result is immediate where $p_1$ is as given in Lemma~\ref{forbid_ball_size}. Let $p\in[p_1, 1/2)$. Consider $G_p(x)$ over the interval $[0, \frac{3}{2}p-\frac{1}{2})$ where $G_p(x)$ is as given in Lemma~\ref{forbid_ball_size}. It is easy to see that, over this interval, $G_p(x)$ is continuous and strictly increasing. Hence, we can define and inverse function $G^{-1}_p(y)$ that maps the range of $G_p(.)$ over $[0, \frac{3}{2}p-\frac{1}{2})$ to $[0, \frac{3}{2}p-\frac{1}{2})$. Note that $G^{-1}_p(y)$ is also continuous and strictly increasing over this interval (i.e., over the image of $[0, \frac{3}{2}p-\frac{1}{2})$ under $G_p(.)$). In this manner, $r_{\delta, \eta}(p)$, as given in Lemma~\ref{forbid_ball_size}, is indeed $G^{-1}_{p}\left(-(\delta+\eta)\right)$. Hence, by the continuity of $G^{-1}_p$, we have
\begin{align}
\lim_{\delta+\eta\rightarrow 0}r_{\delta, \eta}(p)&=\lim_{\delta+\eta\rightarrow 0}G^{-1}_p\left(-(\delta+\eta)\right)=G^{-1}_p(0)=r(p)\nonumber
\end{align}
where $r(p)$ is as given in Theorem~\ref{lower_bound}. This completes the proof.

\vspace{1cm}
\subsection{Proof of Lemma~\ref{good_rand_code}:}
We will first prove that, with probability $1-e^{-2^{\Omega(n)}}$ over the choice of $\enc$, the code associated with $\enc$ is $\frac{\eta}{4}$-good for a fixed pair $(u,~\by_1)$. Then, we conclude the proof by applying the union bound taken over all such pairs.

Let $q\in[0,~\min(R, p)]$. Let $u\in\mU$ and $\by_1\in\{0,1,\wedge\}^{R n}$ such that $\by_1$ is the resulting vector after erasing some $qn$ bits of $u$. Let $\sufx(u)$ denote the set of suffixes of the codewords that correspond to all the messages $u'\neq u$. That is, $\sufx(u)=\{\bbx_2(u',s'):~u'\in\mU\setminus\{u\},~s'\in\ms\}$. Note, by the choice of the code, $\sufx(u)$ is a set of independent and uniformly distributed random variables over $\{0,1\}^{(1-R)n}$. Now, conditioned on the value of local randomness $S=s$, one can think of two sources of randomness in the choice of $\enc$, namely, $\bbx_2(u,s)$ (the suffix of the actual codeword) and $\sufx(u)$ (the list of suffixes of codewords corresponding to all $u'\neq u$).
Note that the event $\ei$ (defined in (\ref{ei})) depends on both sources of randomness in the choice of $\enc$ as well as the randomness due to the choice of $S$. In fact, it can be, equivalently, written as
\begin{align}
\ei\triangleq& \bigg\{\exists~\bbx_2(u',s')\in\sufx(u):\left(u',~\bbx_2(u',s')\right)\in B_{R}^{p,q}\left(\by_1,~\bbx_2(u, S)\right)\bigg\}\nonumber
\end{align}
The probability of $\ei$ taken over the choice of $S$, denoted by $\peI$, is now a random variable since it depends on the choice of $\enc$, namely, it depends on both $\bbx_2(u,S)$ and $\sufx(u)$.

Our goal is to show that $\pr\left(\peI\leq 2^{-\frac{\eta}{4}n}\right)\geq 1-e^{-2^{\Omega(n)}}$ where the outer probability is over the choice of the code (that is, over $\bbx_2(u,S)$ and $\sufx(u)$). To do this, we will show the existence of a subset $\wellb$ of the set of all the possible realizations of $\sufx(u)$, such that, for sufficiently small $\eta>0$, we have
\begin{align}
&\pr\left(\sufx(u)\in\wellb\right)\geq 1-e^{-2^{\Omega(n)}}\nonumber
\end{align}
and
\begin{align}
\pr\left(\peI\leq 2^{-\frac{\eta}{4}n}~\vert~\sufx(u)\in\wellb\right)&\geq 1-e^{-2^{\Omega(n)}}\nonumber
\end{align}



Now, for every $\bz\in\{0,1\}^{(1-R)n}$, define
\begin{align}
&L_{\bz}\left(\by_1,~\sufx(u)\right)\triangleq\bigg\vert\bigg\{\bbx_2(u',s')\in\sufx(u): \left(u',~\bbx_2(u',s')\right)\in B_{R}^{p,q}\left(\by_1,~\bz\right)\bigg\}\bigg\vert\nonumber
\end{align}
Let $\mathbf{1}(.)$ denote the indicator function that takes value $1$ whenever its argument is true and $0$ otherwise. Note that $\ei$ is equivalent to the event that \mbox{$\big\{L_{\bbx_2(u,S)}\left(\by_1,~\sufx(u)\right)\geq 1\big\}$.} Hence, conditioned on $\{\bbx_2(u,s)=\bx_2(u,s),~\sufx(u)=\bc(u):s\in\ms\}$ (that is, for a fixed set of suffixes of all the codewords), we have
\begin{align}
\peI&=\sum_{s\in\ms}\pr(S=s)\mathbf{1}\left(L_{\bx_2(u,s)}\left(\by_1,~\bc(u)\right)\geq 1\right)\nonumber\\
&=\frac{1}{2^{\delta n}}\sum_{s\in\ms}\mathbf{1}\left(L_{\bx_2(u,s)}\left(\by_1,~\bc(u)\right)\geq 1\right)\nonumber
\end{align}

A crucial part in the proof is to obtain, for every $s\in\ms$, an upper bound on
$$E\left[L_{\bbx_2(u,s)}\left(\by_1,\sufx(u)\right)~\Big\vert~\sufx(u)=\bc(u)\right]$$
for any realization $\bc(u)$ of $\sufx(u)$ that lies in a set (denoted as $\wellb$) of an overwhelming probability over the choice of $\sufx(u)$. We proceed as follows. First, we find a set $\wellb$ of realizations of $\sufx(u)$ for which the sum $\sum_{\bz\in\{0,1\}^{(1-R)n}}L_{\bz}\left(\by_1, \sufx(u)\right)$ is not too far from its expectation (over $\sufx(u)$) and show that such set has an overwhelming probability (over the choice of $\sufx(u)$). Using this we then obtain an upper bound on the conditional expectation $E\left[L_{\bbx_2(u,s)}\left(\by_1,\sufx(u)\right)~\Big\vert~\sufx(u)=\bc(u)\right]$ (over $\bbx_2(u,s)$) for every $s\in\ms$ and every $\bc(u)\in\wellb$. Finally, we use this to show that, with overwhelming probability (over the choice of the suffixes $\{\bbx_2(u,s):~s\in\ms\}$), we have $\peI\leq 2^{-\frac{\eta}{4}n}$.

For every $u'\in\mU\setminus\{u\},~s'\in\ms$, define
\begin{align}
&V_{u',s'}\left(\by_1,~\sufx(u)\right)\triangleq\bigg\vert\bigg\{\bz\in\{0,1\}^{(1-R)n}:~\enc(u',s')\in B^{p,q}_{R}(\by_1, \bz)\bigg\}\bigg\vert\label{V_u'}
\end{align}
Now, observe that
\begin{align}
\hspace{-2cm}\sigma\left(\by_1, \sufx(u)\right)&\triangleq\sum_{\bz\in\{0,1\}^{(1-R)n}}L_{\bz}\left(\by_1, \sufx(u)\right)\label{equal_sums}\\
&=\sum_{u'\in\mU\setminus\{u\}}\sum_{s'\in\ms}V_{u',s'}\left(\by_1,~\sufx(u)\right)\nonumber\\
&=\bigg\vert\bigg\{(u',s',\bz)\in\left(\mU\setminus\{u\}\right)\times\ms\times\{0,1\}^{(1-R)n}:\enc(u',s')\in B^{p,q}_{R}(\by_1, \bz)\bigg\}\bigg\vert\nonumber
\end{align}
On the other hand, \mbox{$\big\{V_{u',s'}\left(\by_1, \sufx(u)\right): u'\in\mU\setminus\{u\},$} $s'\in\ms\big\}$ are independent random variables. Moreover, for every $u'\in\mU\setminus\{u\},~s'\in\ms$
\begin{align}
E[V_{u',s'}\left(\by_1,~\sufx(u)\right)]&\leq 2^{(1-R)n}\frac{B^{p,q}_{R}}{2^n-2^{(1-R)n}}\leq 2\cdot2^{(1-R)n}\frac{B^{p,q}_{R}}{2^n}=\frac{2B^{p,q}_{R}}{2^{R n}}\nonumber
\end{align}
for sufficiently large $n$. For all \mbox{$(u',s')\in\left(\mU\setminus\{u\}\right)\times\ms$}, let $\tilde{V}_{u',s'}\left(\by_1,~\sufx(u)\right)\triangleq\frac{V_{u',s'}(\by_1,~\sufx(u))}{B^{p,q}_{R}}$. Clearly, from (\ref{V_u'}), we have $\tilde{V}_{u',s'}\left(\by_1,~\sufx(u)\right)\leq 1$ for all \mbox{$(u',s')\in\left(\mU\setminus\{u\}\right)\times\ms$}.


Let
\begin{align}
\tilde{\sigma}\left(\by_1, \sufx(u)\right)&\triangleq\sum_{u'\in\mU\setminus\{u\}}\sum_{s'\in\ms}\tilde{V}_{u',s'}\left(\by_1,~\sufx(u)\right)\label{sigma-tilde}
\end{align}
Hence, we have
\begin{align}
\pr\left(\sigma\left(\by_1, \sufx(u)\right)\geq 2^{n\delta+2}B^{p,q}_{R}\right)&=\pr\left(\tilde{\sigma}\left(\by_1, \sufx(u)\right)\geq 2^{n\delta+2}\right)\label{sum_of_L_bx_2_a}\\
&\leq\pr\left(\tilde{\sigma}\left(\by_1, \sufx(u)\right)\geq 4\frac{2^{nR_{\delta,\eta}(p)}-2^{\delta n}}{2^{R n}}\right)\nonumber\\
&\leq e^{-2^{\Omega(n)}}\label{sum_of_L_bx_2_e}
\end{align}
where (\ref{sum_of_L_bx_2_a}) follows from (\ref{equal_sums}), (\ref{sigma-tilde}), and the definition of $\tilde{V}_{u',s'}\left(\by_1, \sufx(u)\right)$ above, and (\ref{sum_of_L_bx_2_e}) follows from Chernoff-Hoeffding bound \cite{hoeffding} restated in the following lemma.

\begin{lemma}\label{hoeff}(Chernoff-Hoeffding)
Let $X_1, X_2, ..., X_N$ be independent random variables taking values in $[0,~1]$ with expectation at most $\mu$. Then,
\begin{align}
\pr\left(\sum_{i=1}^N X_i\geq 2\mu N\right)&\leq e^{-\Omega(\mu N)}\nonumber
\end{align}
\end{lemma}
\vspace{0.6cm}

Now, we define $\wellb$ as the set of all realizations $\bc(u)$ of $\sufx(u)$ such that
\begin{align}
\sum_{\bz\in\{0,1\}^{(1-R)n}}L_{\bz}\left(\by_1, \bc(u)\right)&\leq 2^{n\delta+2}B^{p,q}_{R}\nonumber
 \end{align}
Thus, we have
\begin{align}
\pr\left(\sufx(u)\in\wellb\right)\geq 1-e^{-2^{\Omega(n)}}\nonumber
\end{align}

Next, we will show that, conditioned on $\big\{\sufx(u)\in\wellb\big\}$, we have
$$\peI=\frac{1}{2^{\delta n}}\sum_{s\in\ms}\mathbf{1}\left(L_{\bbx_2(u,s)}\left(\by_1, \sufx(u)\right)\geq 1\right)\leq 2^{-\frac{\eta}{4}n}$$
with overwhelming probability over the choice of $\{\bbx_2(u,s):~s\in\ms\}$. To do this, we first bound the conditional expectation $E\left[L_{\bbx_2(u,s)}\left(\by_1,\sufx(u)\right)~\Big\vert~\sufx(u)=\bc(u)\right]$ (over $\bbx_2(u,s)$) for every $s\in\ms$ and every $\bc(u)\in\wellb$.

Observe that, for every $s\in\ms$ and every $\bc(u)\in\wellb$,
\begin{align}
E\left[L_{\bbx_2(u,s)}\left(\by_1,\sufx(u)\right)~\Big\vert~\sufx(u)=\bc(u)\right]&\leq\frac{1}{2^{(1-R)n}}2^{n\delta+2}B^{p,q}_{R}\label{expected_N_u_c}\\
&\leq 4\cdot2^{-\frac{\eta}{2}n}\label{expected_N_u_d}
\end{align}
where (\ref{expected_N_u_c}) follows from the fact that, conditioned on $\big\{\sufx(u)=\bc(u)\big\}$ where $\bc(u)\in\wellb$, we must have  $\sum_{\bz\in\{0,1\}^{(1-R)n}}L_{\bz}\left(\by_1, \sufx(u)\right)\leq 2^{n\delta+2}B^{p,q}_{R}$, and (\ref{expected_N_u_d}) follows from Lemma~\ref{forbid_ball_size}. It follows that, for every $s\in\ms$ and every $\bc(u)\in\wellb$, we must have
\begin{align}
E\left[\mathbf{1}\left(L_{\bbx_2(u,s)}\left(\by_1,\sufx(u)\right)\geq 1\right)~\Big\vert~\sufx(u)=\bc(u)\right]&\leq E\left[L_{\bbx_2(u,s)}\left(\by_1,\sufx(u)\right)~\Big\vert~\sufx(u)=\bc(u)\right]\nonumber\\
&\leq 4\cdot2^{-\frac{\eta}{2}n}\label{bd-expect}
\end{align}

Moreover, observe that, conditioned on $\big\{\sufx(u)=\bc(u)\big\}$ where $\bc(u)\in\wellb$, the collection $\{L_{\bbx_2(u,s)}\left(\by_1,\bc(u)\right):~s\in\ms\}$ is independent and identically distributed. Recall that $\peI=\frac{1}{2^{\delta n}}\sum_{s\in\ms}\mathbf{1}\left(L_{\bbx_2(u,s)}\left(\by_1, \sufx(u)\right)\geq 1\right)$. Thus, we have
\begin{align}
&\pr\left(\peI\geq 2^{-\frac{\eta}{4}n}~\Big\vert~\sufx(u)\in\wellb\right)\nonumber\\
=&\sum_{\bc(u)\in\wellb}\pr\left(\peI\geq 2^{-\frac{\eta}{4}n}~\Big\vert~\sufx(u)=\bc(u)\right)\pr\left(\sufx(u)=\bc(u)~\Big\vert~\sufx(u)\in\wellb\right)\nonumber\\
\leq& e^{-2^{\Omega(n)}}\sum_{\bc(u)\in\wellb}\pr\left(\sufx(u)=\bc(u)~\Big\vert~\sufx(u)\in\wellb\right)\label{pr_sum_I_u_b}\\
=&e^{-2^{\Omega(n)}}\nonumber
\end{align}
where (\ref{pr_sum_I_u_b}) follows from (\ref{bd-expect}) and Chernoff-Hoeffding bound (Lemma~\ref{hoeff}).

Thus, we finally get
\begin{align}
\pr\left(\peI\leq 2^{-\frac{\eta}{4}n}\right)&\geq\pr\left(\peI\leq 2^{-\frac{\eta}{4}n}~\Big\vert~\sufx(u)\in\wellb\right)\pr\left(\sufx(u)\in\wellb\right)\nonumber\\
&\geq\left(1-e^{-2^{\Omega(n)}}\right)\left(1-e^{-2^{\Omega(n)}}\right)\nonumber\\
&=1-e^{-2^{\Omega(n)}}\label{pr_goodness_3}
\end{align}

So far, we have shown that, with overwhelming probability over the choice of $\enc$, the code associated with $\enc$ is $\frac{\eta}{4}$-good with respect to a fixed pair $(u,~\by_1)$. To complete the proof, we apply the union bound over all possible pairs $(u,~\by_1)\in\mU\times\{0,1,\wedge\}^{Rn}$ such that $\by_1$ has at most $\min(p, R) n$ erasures. Note that due to the doubly exponential probability profile of (\ref{pr_goodness_3}), we still attain $\frac{\eta}{4}$-goodness with probability at least $1-e^{-2^{\Omega(n)}}$ after applying the union bound.

\bibliographystyle{plain}
\bibliography{ref_coding_causal_ch,avc}

\end{document}